\documentclass[pra,superscriptaddress,preprint]{revtex4}
\usepackage[english]{babel}
\usepackage{graphicx}
\usepackage{dcolumn}
\usepackage{bm}
\usepackage{dsfont}
\usepackage{enumerate}
\usepackage{color}
\usepackage{mathrsfs}
\usepackage{braket}
\usepackage{amssymb,amsmath,amsthm}
 \newtheorem{claim}{Claim}
 \newtheorem{remark}{Remark}
\usepackage{epsfig}
\begin{document}

\title{Bounds on skew information and local quantum uncertainty \\ for state conversion}
\author{Liang Qiu}
\affiliation{%
	School of Physics, China University of Mining and Technology, Xuzhou, Jiangsu 221116, China%
	}
\email{lqiu@cumt.edu.cn}
\affiliation{%
    Institute for Quantum Science and Technology, University of Calgary, Alberta T2N 1N4, Canada%
	}

\author{Yu Guo}
\affiliation{%
	Institute for Quantum Science and Technology, University of Calgary, Alberta T2N 1N4, Canada%
	}
\affiliation{%
	Institute of Quantum Information Science, Shanxi Datong University, Datong, Shanxi 037009, China%
	}

\author{Barry C. Sanders}
\affiliation{%
    Institute for Quantum Science and Technology, University of Calgary, Alberta T2N 1N4, Canada%
    }
\affiliation{%
    Program in Quantum Information Science,
    Canadian Institute for Advanced Research,
    Toronto, Ontario M5G 1Z8, Canada%
    }
\affiliation{%
    Hefei National Laboratory for Physical Sciences at Microscale,
    University of Science and Technology of China, Hefei, Anhui 230026, China
    }
\affiliation{%
    Shanghai Branch,
    CAS Center for Excellence and Synergetic Innovation Center
        in Quantum Information and Quantum Physics,
    University of Science and Technology of China, Shanghai 201315, China
    }

\date{\today}
\begin{abstract}
We establish strict upper bounds on local quantum uncertainty (LQU) and skew information associated with state conversion via certain quantum channels. Specifically, we obtain a bound on the achievable LQU for bipartite channels whose Kraus operators commute with nondegenerate von Neumann measurements on the first subsystem, and this LQU bound is expressed in terms of the skew information for the first subsystem. Furthermore, we establish a bound on the skew information of one subsystem obtained from any initial bipartite state subject to any quantum steering channel, and this bound is expressed in terms of the LQU for the initial joint-system state. Our two claims show that state conversion has fundamental limitations relating LQU with skew information.
\end{abstract}

\maketitle

\section{Introduction}

Skew information, motivated by the study of quantum measurements in the presence of a conserved quantity~\cite{Araki1960,Yanase1961}, was introduced by Wigner and Yanase in 1963~\cite{Wigner1963} and was originally used to describe the information content of mixed states.\ Luo demonstrated that the statistical idea underlying skew information is the Fisher information, used for the theory of statistical estimation\ \cite{Luo2003}, while Fisher information is not only a key notion of statistical inference\ \cite{Fisher1925,Cramer1974} but also plays an increasing role in informational treatments of physics\ \cite{Helstrom1976,Holevo1982,Frieden1995,Luo2002,Frieden2004}. Based on skew information, an intrinsic measure for synthesizing quantum uncertainty of a mixed state has been proposed\ \cite{Luo2006}. The measure for correlations in terms of skew information is also given, and the evaluation of the measure does not involve any optimization, in sharp contrast to the case for entanglement and discord measures~\cite{Luo2012}. As for correlations, Girolami et al.\ defined and investigated a class of measures of bipartite quantum correlations of the discord type\ \cite{Modi2012} through local quantum uncertainty (LQU)~ \cite{Girolami2013}, which is also obtained from the skew information.\ Skew information is also proven to be an asymmetry monotone in Refs.\ \cite{Girolami2014,Marvian2014}. Moreover, Girolami originally proposed the skew information as a coherence monotone~\cite{Girolami2014}; however, such a quantity may increase under the action of incoherent operations~\cite{Du2015}. In other words, skew information violates monotonicity, which is one of the postulates that any quantifier of coherence should fulfill~\cite{Baumgratz2014}.

Coherence of a single system can be traded for quantum correlations.\ The relation between coherence and entanglement is studied in Refs.\ \cite{Streltsov2015,Yao2015,Davide2015,Chitambar2016}. The link between specific coherence and discord-type measures has also been reported in Refs.\ \cite{Xi2015,Yao2015,Datta2012}. The interplay between coherence and quantum discord in multipartite systems has been investigated in Ref.~\cite{Ma2016}. Motivated by these results, we now investigate bounds relating skew information to LQU.
Specially, we prove that the LQU created between a single partite and an ancilla by quantum operations is bounded above by the initial skew information of the single system.

Quantum steering is a process by which Alice can steer the quantum state of Bob by her local selective measurement if they initially share a correlated quantum system. It is a kind of nonlocal correlation introduced by Schr\"{o}dinger~\cite{Schrodinger1935,Schrodinger1936} to reinterpret the Einstein-Podolsky-Rosen (EPR) paradox~\cite{Einstein1935}. According to Schr\"{o}dinger, entanglement between two subsystems in a bipartite state is the vital ingredient in quantum steering. EPR steering has received much attention both theoretically and experimentally \cite{Wiseman2007,Saunders2010,Handchen2012,Skrzypczyk2014}. Quantum steering is intimately connected to remote state preparation \cite{Babichev2004}. The power of Alice's local probabilistic measure to create coherence on Bob's side is also investigated~\cite{Hu2016}. Moreover, the complementarity relations between coherence measured on mutually unbiased bases using various coherence measures have been obtained~\cite{Monda2015}.

The converse procedure of converting local coherence to quantum correlations, i.e., to extract local coherence from a spatially separated yet quantum correlated bipartite system, is of importance.\ Chitambar et al.\ introduced and studied the task of assisted coherence distillation, where local quantum-incoherent operations and classical communication are employed\ \cite{ChitambarStreltsov2016}.\ Coherence can also be extracted from measurement-induced disturbance~\cite{Hu2015}, while the latter characterizes the quantumness of correlations~\cite{Luo2008}.
Motivated by these results, we also investigate bounds on skew information with respect to LQU
with these bounds arising from quantum steering channels:
we find that LQU is the upper bound of skew information.

The paper is organized as follows. In Sec.~II, we introduce the definitions of skew information and LQU.
Bounding skew information by LQU is investigated in Sec.~III.
Subsequently, we investigate a bound on skew information from LQU for the process of quantum steering in Sec.~IV. We conclude in Sec.~V.

\section{skew information and local quantum uncertainty}
Skew information is
\begin{equation}
	I(\rho,X):=-\frac{1}{2}\operatorname{Tr}[\rho^{1/2},X]^2,
\end{equation}
where $\operatorname{Tr}$ denotes the trace, and $[\bullet,\bullet]$ denotes the commutator. $\rho$ is a general quantum state and~$X$ is an observable, which is a self-adjoint operator, and~$X$ can be a Hamiltonian.\ If $[\rho,X]=0$, then~$X$ is conserved~\cite{Luo2012}.\ The skew information provides an alternative measure of the information content for $\rho$ with respect to observables not commuting with the conserved quantity~$X$~\cite{Wigner1963}.

Skew information has the following nice properties.

(1) For pure states ($\rho^2=\rho$), $I(\rho,X)$ reduces to the conventional variance $V(\rho,X):=\operatorname{Tr}\rho X^2-(\operatorname{Tr}\rho X)^2$ and~$I(\rho,X)$ satisfies $0\leq I(\rho,X)\leq V(\rho,X)$.

(2) $I(\rho,X)$ is convex, which means the skew information decreases when several states are mixed~\cite{Wigner1963,Lieb1973}:
\begin{equation}
I\left(\sum_i\alpha_i\rho_i,X\right)\leq\sum_i\alpha_iI(\rho_i,X),\ \alpha_i\in\mathds{R}.
\end{equation}
Here the probability distribution $\{\alpha_i\}$ satisfies $\sum_i\alpha_i=1$ and $0\leq \alpha_i\leq 1$. On the contrary, the variance $V(\rho,X)$ is concave in $\rho$.

(3) In the Hilbert space $\mathscr{H}_A\otimes \mathscr{H}_B$ of a composite system $AB$, the skew information of quantum state $\rho_{AB}$ and that of the reduced density matrix
$\rho_A:=\operatorname{Tr}_B\rho_{AB}$ have the relation~\cite{Lieb1973}
\begin{equation}
	I\left(\rho_{AB},X_A\otimes \mathds{I}_B\right)\geq I(\rho_A,X_A)
\end{equation}
for any observable $X_A$ in $\mathscr{H}_A$. Here $\mathds{I}_B$ is the identity operator in $\mathscr{H}_B$.

The state $\rho$ and the observable~$X$ fix the skew information $I(\rho,X)$. In order to eliminate the observable on skew information, and get an intrinsic quantity capturing the information content of $\rho$, Luo introduced the average~\cite{Luo2006,LiX2011}
\begin{equation}
	Q\left(\rho\right)
		:=\sum\limits_{i=1}^{n^2}I(\rho,X^i),
\end{equation}
where $\{X^i\}$ constitutes an orthonormal basis for an $n^2$-dimensional Hilbert space $\mathcal{L}(\mathscr{H})$ of all observables on $n$-dimensional quantum system with the Hilbert-Schmidt inner product $\langle X,Y\rangle=\operatorname{Tr}XY\in\mathds{R}$. Then $Q\left(\rho\right)$ depends only on the quantum state $\rho$ and is independent of choice of the orthonormal basis $\{X^i\}$. $Q\left(\rho\right)$ is considered not only to be a measure of information content of $\rho$, but also a measure of quantum uncertainty~\cite{Luo2006}.
The global information content of $\rho_{AB}$ in terms of local observables of $n$-dimensional quantum system $A$ is
\begin{equation}
Q_A(\rho_{AB})=\sum\limits_{i=1}^{n^2}I(\rho_{AB},X_A^i\otimes \mathds{I}_B),\ X_A^i\in\mathcal{L}(\mathscr{H}_A).
\end{equation}

Girolami et al.\ defined LQU as the minimum skew information achievable on local von Neumann measurement of a subsystem\ \cite{Girolami2013}.\ Specifically, for a bipartite quantum state $\rho_{AB}$, we consider the local observable $K^\Lambda=K_A^\Lambda\otimes \mathds{I}_B$ such that $K_A^\Lambda$ is a Hermitian operator on subsystem $A$ with nondegenerate spectrum $\Lambda$.\ Optimized over all local observables on $A$ with nondegenerate spectrum $\Lambda$, the LQU with respect to subsystem $A$ is
\begin{equation}
\mathcal{U}_A^\Lambda=\min\limits_{K^\Lambda_A}I(\rho_{AB},K^\Lambda).
\end{equation}
LQU is shown to be a full-fledged measure of bipartite quantum correlations, and it can be evaluated analytically for the case of bipartite $2\times d$ systems \cite{Girolami2013}.

\section{Bounds for skew information and local quantum uncertainty}
In terms of the Kraus representation, a completely positive trace-preserving map~$\Phi$ transforms $\rho$ into another state $\Phi\left(\rho\right)$ by
\begin{equation}
	\Phi\left(\rho\right)=\sum\limits_{j}E_{j}\rho E_{j}^{\dag}
\end{equation}
with the unit trace condition $\operatorname{Tr}\Phi\left(\rho\right)=1$ leading to $\sum_{j}E_{j}^{\dag}E_{j}=\mathds{I}$. Now we present the first claim of the paper.

\begin{claim}
\label{claim:LQUSI}
If $K_A^\Lambda$ is any von Neumann measurement acting on $\mathscr{H}_A$ with nondegenerate spectrum $\Lambda$, the LQU created between a state $\rho_A$ in system $A$ and an arbitrary ancilla $\tau_B$ in system $B$ by quantum operation~$\Phi$ with Kraus operators $E_{j}$s commuting with $K_A^\Lambda\otimes \mathds{I}_B$, i.e., $[E_{j},K_A^\Lambda\otimes \mathds{I}_B]=0$, is bounded above by the skew information of $\rho_A$ in terms of $K_A^\Lambda$:
\begin{equation}
	\max\limits_{\tau_B}\mathcal{U}_A^\Lambda(\Phi(\rho_A\otimes \tau_B))\leq I(\rho_A,K_A^\Lambda).\label{result1}
\end{equation}
\end{claim}
\begin{proof}
First of all, according to the results given in Ref.~\cite{Luo2012},
i.e.,
\begin{equation}
	I(\rho_A\otimes\tau_B,X_A\otimes\mathds{I}_B)=I(\rho_A,X_A),
\end{equation}
we have
\begin{equation}
	I(\rho_A,K_A^\Lambda)=I(\rho_A\otimes \tau_B,K_A^\Lambda\otimes \mathds{I}_B).
\end{equation}
Based on Theorem 1 in Ref.~\cite{Luo2007}, which shows $I(\Phi\left(\rho\right),X)\leq I(\rho,X)$ if the quantum operation~$\Phi$ does not disturb the observable~$X$ in the sense that all $E_{j}$s commuting with~$X$, we obtain
\begin{equation}
	I(\rho_A\otimes \tau_B,K_A^\Lambda\otimes\mathds{I}_B)\geq I(\Phi(\rho_A\otimes \tau_B),K_A^\Lambda\otimes \mathds{I}_B)
\end{equation}
because we have assumed that
\begin{equation}
	[E_{j},K_A^\Lambda\otimes \mathds{I}_B]=0\;\forall j.
\end{equation}
Clearly
\begin{equation}
	I(\Phi(\rho_A\otimes \tau_B),K_A^\Lambda\otimes \mathds{I}_B)\geq \min\limits_{K_A^\Lambda}I(\Phi(\rho_A\otimes \tau_B),K_A^\Lambda\otimes \mathds{I}_B),
\end{equation}
the right-hand side of which is just the LQU of the state $\Phi(\rho_A\otimes \tau_B)$, i.e.,
\begin{equation}
	\mathcal{U}_A^\Lambda(\Phi(\rho_A\otimes \tau_B)).
\end{equation}
\end{proof}
\begin{remark}
Claim~\ref{claim:LQUSI}
demonstrates that the information content of the mixed state~$\rho_A$ for subsystem~$A$
bounds the discord-type quantum correlations for the joint AB system for any tensor product state of~$\rho_A$ in subsystem A and arbitrary state in subsystem~$B$.
\end{remark}
Now we investigate bounds on skew information in terms of LQU for quantum steering channels.
Alice and Bob are assumed to share a quantum correlated state $\rho_{AB}$ initially. Bob's state is steered to
\begin{equation}
	\rho_B^i
		=\frac{\bra{\theta_A^i}\rho_{AB}\ket{\theta_A^i}}{p^i},\;
	p^i
		=\operatorname{Tr}[\rho_{AB}(\theta_A^i\otimes \mathds{I}_B)]
\end{equation}
after Alice implements local projective measurement
\begin{equation}
	\theta_A^i=\ket{\theta^i}_A\bra{\theta^i},\;
	(i=0,1,\cdots,n_A-1),
\end{equation}
where~$n_A$ is the dimension of the subsystem $A$. If $K_B^\Lambda$ is just the von Neumann measurement with nondegenerate spectrum $\Lambda$ in LQU, i.e.,
\begin{equation}
	\mathcal{U}^\Lambda_B(\rho_{AB})=\min\limits_{K_B^\Lambda}I\left(\rho_{AB},\mathds{I}_A\otimes K_B^\Lambda\right),
\end{equation}
we define the steering-induced skew information in terms of $K_B^\Lambda$ as
\begin{equation}
	\bar{I}(\rho_B)
		:=\max\limits_{\Theta_A}\sum\limits_ip^iI(\rho_B^i,K_B^\Lambda).
\end{equation}
Here the maximization is taken over all of Alice's projective measurement basis $\Theta_A=\{\theta_A^i\}$
$(i=0,1,\cdots,n_A-1)$.

\begin{claim}
\label{claim:SILQU}
The steering-induced skew information is bounded above by LQU with respect to the subsystem $B$; i.e.,
\begin{equation}
	\bar{I}(\rho_B)\leq \mathcal{U}^\Lambda_B(\rho_{AB}).
\end{equation}
\end{claim}
\begin{proof}
In order to prove the result, we first note that
\begin{equation}
	\sum_ip^iI(\rho_B^i,K_B^\Lambda)=\sum_ip^iI(\theta_A^i\otimes \rho_B^i,\mathds{I}_A\otimes K_B^\Lambda)
\end{equation}
due to the fact
\begin{equation}
	I(\rho_A\otimes\tau_B,\mathds{I}_A\otimes K_B^\Lambda)
		=I(\tau_B,K_B^\Lambda),
\end{equation}
which has been proved in Eq.~(8) in Ref.~\cite{Luo2012}.
Based on the fact that
\begin{equation}
	\theta_A^i\otimes \rho_B^i
		=\frac{\theta_A^i\rho_{AB}\theta_A^i}{p^i},
\end{equation}
we obtain
\begin{equation}
	\sum_ip^iI(\theta_A^i\otimes \rho_B^i,\mathds{I}_A\otimes K_B^\Lambda)=\sum_ip^iI((\theta_A^i\rho_{AB}\theta_A^i)/p^i,\mathds{I}_A\otimes K_B^\Lambda).
\end{equation}
In Ref.~\cite{Luo2012}, the result
\begin{equation}
	I((\varepsilon_A\otimes \mathds{I}_B)\rho_{AB},\mathds{I}_A\otimes K_B^\Lambda)\leq I\left(\rho_{AB},\mathds{I}_A\otimes K_B^\Lambda\right)
\end{equation}
with~$\varepsilon_A$ being an operation on the state space of $H_A$ has been proved.
Therefore,
\begin{equation}
	\sum_ip^iI((\theta_A^i\rho_{AB}\theta_A^i)/p_i,\mathds{I}_A\otimes K_Ba)
		\leq \sum_ip^iI\left(\rho_{AB},\mathds{I}_A\otimes K_B^\Lambda\right),
\end{equation}
the right-hand side of which yields
\begin{equation}
	I\left(\rho_{AB},\mathds{I}_A\otimes K_B^\Lambda\right)\sum_ip^i
		=I\left(\rho_{AB},\mathds{I}_A\otimes K_B^\Lambda\right),
\end{equation}
i.e., LQU in terms of the observable of the subsystem $B$.
\end{proof}

\begin{remark}
Claim~\ref{claim:SILQU}
indicates that the discord-type quantum correlations for any state of the joint $AB$ system bounds the information content for a subsystem in the procedure of quantum steering.
\end{remark}

Obviously, the definition of the steering-induced skew information given above depends on the local von Neumann measurement. In order to eliminate it, we introduce the average steering-induced skew information as
\begin{equation}
\bar{Q}(\rho_B)=\max\limits_{\Theta_A}\sum\limits_{i,j}p^iI(\rho_B^i,X_B^{j}),
\end{equation}
where ${X_B^{j}}$ constitutes an orthonormal basis for the Hilbert space $\mathcal{L}(\mathscr{H}_B)$ of all observables on subsystem $B$. $\bar{Q}(\rho_B)$ depends only on Alice's projective measurement and the state shared by Alice and Bob, while is independent of the observable.

According to the proof procedure for Claim 2, we have
\begin{align}
	\sum_{j}I(\rho_B^i,X_B^{j})
		=&\sum_{j}I(\theta_A^i\otimes\rho_B^i,\mathds{I}_A\otimes X_B^{j})
			\nonumber\\
		=&\sum_{j}I\left((\theta_A^i\rho_{AB}\theta_A^i)/p^i,\mathds{I}_A\otimes X_B^{j}\right)
			\nonumber\\
		\leq &\sum_{j}I(\rho_{AB},\mathds{I}_A\otimes X_B^{j})
\end{align}
Subsequently, we have
\begin{equation}
	\bar{Q}(\rho_B)
		\leq\max\limits_{\Theta_A}\sum_ip^i\sum_{j}I(\rho_{AB},\mathds{I}_A\otimes X_B^{j})
		=\sum_{j}I(\rho_{AB},\mathds{I}_A\otimes X_B^{j})=Q_B(\rho_{AB}).
\end{equation}
\begin{remark}
The average steering-induced skew information is bounded above by the global information content of $\rho_{AB}$ with respect to the local observables of the subsystem $B$.
\end{remark}

\section{Conclusions}
Resource interconvertibility,
namely,
trading one resource for another,
motivates research on the transformation between quantum coherence and either entanglement or quantum correlations.
Here we investigate interconvertibility between skew information and LQU.
The former could be viewed as the measure of the informational content of mixed states, an asymmetry monotone, and even as a measure of quantum coherence.
The latter measure, LQU, is a full-fledged measure of bipartite quantum correlations.

Our work establishes bounds on skew information and LQU associated with state conversion via certain quantum channels.
In particular, by using bipartite channels whose Kraus operators commute with nondegenerate von Neumann measurements on the first subsystem,
discord-type quantum correlations for the joint system are bounded above by skew information of the subsystem.

On the contrary, for the procedure of quantum steering, LQU of the initially shared state is the upper bound of the steering induced skew information.
Furthermore, for the case of being independent of the observables, the average steering-induced skew information of the subsystem is bounded above by the global information content of the bipartite system with respect to the local observables.

\acknowledgements
L. Q.\ acknowledges support from the Fundamental Research Funds for the Central Universities under Grant No.\ 2015QNA44.
B. C. S.\ appreciates financial support from NSERC,
Alberta Innovates and China's 1000 Talent Plan.
This work was completed while Y. G. was visiting the Institute for Quantum Science and Technology at the University of Calgary under the support of the China Scholarship Council.

\end{document}